\documentclass[11pt]{article}
\usepackage[T1]{fontenc}
\usepackage{lmodern}
\usepackage{amsmath,amssymb,amsthm,mathtools}
\usepackage{bm}
\usepackage{booktabs}
\usepackage{float}
\usepackage[margin=1in]{geometry}
\usepackage{microtype}
\usepackage{enumitem}
\usepackage{xcolor}
\usepackage{soul}
\usepackage[colorlinks=true,linkcolor=blue!65!black,citecolor=blue!65!black,
urlcolor=blue!65!black]{hyperref}
\newcommand{\K}{\mathrm{K}}
\newcommand{\RC}{\mathrm{RC}}
\newcommand{\Len}{\mathrm{L}}
\newcommand{\B}{\{0,1\}}
\newcommand{\ind}{\mathbf{1}}
\newcommand{\aeq}{\overset{+}{=}}
\newcommand{\aleq}{\overset{+}{\leq}}
\newcommand{\ageq}{\overset{+}{\geq}}
\newcommand{\Ess}{\mathcal{I}}
\newtheorem{theorem}{Theorem}
\newtheorem{proposition}[theorem]{Proposition}
\newtheorem{lemma}[theorem]{Lemma}

\begin{document}
\title{Beyond expressiveness in pairwise and higher-order models}
\author{Cyril Rommens$^{1,2}$, Pietro Traversa$^{1,2}$,\\
Guilherme Ferraz de Arruda$^{3}$, and Yamir Moreno$^{1,2}$\\[0.5em]
\small $^1$Institute for Biocomputation and Physics of Complex Systems (BIFI),\\
\small Universidad de Zaragoza, 50018 Zaragoza, Spain\\
\small $^2$Department of Theoretical Physics, University of Zaragoza, 50009 Zaragoza, Spain\\
\small $^3$Instituto de F\'\i sica Gleb Wataghin, Universidade Estadual de Campinas (UNICAMP), Campinas, Brazil\\
\small Corresponding author: \texttt{yamir@unizar.es}}
\date{}
\maketitle

\begin{abstract}
The debate over pairwise and higher-order models is often framed as a contest of expressive power, a framing we consider misleading. A graph equipped with arbitrary multivariate node functions can reproduce the node-level dynamics of a broad class of hypergraph models. However, the existence of such an emulation does not make the grouping information encoded by the hypergraph disappear: it may simply be transferred from the structural description to the dynamical rule. We separate four notions that are frequently conflated in this discussion: structural representation, functional representability, statistical identifiability, and mechanistic adequacy. We then prove that interaction order is a property of the dynamics rather than of its representation: the number of variables that must act jointly in at least one term, however a node's update rule is split into a sum of terms, is the same in every exact representation, so no change of structural language can lower it. The comparison is then formulated as a description-length problem. At the unrestricted algorithmic level, a fixed compiler can redistribute information between structure and rule without imposing more than a constant overhead. Expressiveness alone therefore cannot privilege either graphs or hypergraphs.
Preferences arise only relative to explicit model classes, code families, regularity assumptions, and data. This leads to an operational minimum-description-length criterion that combines the cost of the structure, the cost of the rule conditional on that structure, and imperfect fit. Within that criterion we exhibit an explicit separation: for $M$ disjoint groups of size $k$ the edge list of the clique projection is asymptotically $k-1$ times longer than the hyperedge list it replaces, so the projection encodes the same grouping at a higher cost. Examples involving pairwise diffusion, group-sensitive Boolean dynamics, ecological interactions, and ambiguous projections illustrate graph-preferred, hypergraph-preferred, and observationally unresolved cases; bipartite and multilayer representations give a further case in which cost, fit, and identifiability tie and the choice turns on which entities a formalism posits as primitive. The resulting position is symmetric: higher-order structure should not be inferred from phenomenology alone, and the ability of an unrestricted graph rule to emulate a system is no evidence that a graph is its most parsimonious or scientifically adequate description.
\end{abstract}

\section{Introduction}
\label{sec:introduction}

Networks provide a successful language for describing systems whose behavior depends on relations among their components~\cite{newman2018networks,Strogatz2001,Boccaletti2006}. Their success has also made the boundaries of that language scientifically important. Many data sets record joint events, groups, assemblies, complexes, or co-activations rather than isolated dyadic contacts, which has motivated hypergraphs, simplicial complexes, and related higher-order representations of interaction structure~\cite{Torres2021,Battiston2020,BickGross2023,Bianconi2021book}. The position that graphs should be regarded as a special case of these richer formalisms has been argued at length, from an early perspective~\cite{battiston2021physics} to a recent claim that higher-order interactions shape collective human behavior~\cite{Battiston2025NHB} and a community-wide roadmap treating hypergraphs and simplicial complexes as the natural starting point for future work~\cite{AbiadRoadmap2026}. These formalisms have, in turn, prompted a legitimate question: when do they represent additional interaction structure, and when do they merely repackage dynamics that could have been written on a graph?

One influential line of argument answers this question through expressiveness. A graph determines which variables are available to each node, but it need not impose a pairwise-additive update rule. A node function
\begin{equation}
 \dot x_i=f_i(x_i,x_{\partial i})
 \label{eq:general-graph}
\end{equation}
may depend jointly and nonlinearly on the entire neighborhood. If the functions $f_i$ are unrestricted, a graph can therefore emulate the node-level dynamics generated by many hypergraph models. This observation is correct and useful, since it dispels the common but mistaken identification of ``an edge joins two nodes'' with ``all dynamical contributions must be additive and pairwise.'' Recent work has stated the conclusion in its strongest form, describing graphs as maximally expressive and hypergraph models as a restricted subclass~\cite{Peixoto2026}. Other recent contributions have instead asked when higher-order dynamics are reducible to, or transformable into, pairwise models under specified dynamical assumptions and observables~\cite{Lucas2026,Xie2026,Llabres2026}, or have distinguished the order of a hypergraph from the effective order of the dynamics it carries and learned the latter from trajectories~\cite{Neuhauser2024}, and a recent review formulates reducibility itself as a comparison against a declared class of lower-order models, relative to a target and a tolerance~\cite{PerezReche2026}. These results reinforce the need to state the admissible model classes and the level of equivalence at which a claim is made. The work of Llabr\'es \emph{et al.}~\cite{Llabres2026} is instructive for what follows. It shows that a broad class of social impact models on hypergraphs maps exactly onto pairwise dynamics on a weighted projection, with weights that are constant only when the impact function is linear. Otherwise the reduction keeps a pairwise skeleton and the grouping survives in state-dependent weights, a relocation of the kind this article is about.

The question has been taken up independently and almost simultaneously in two recent preprints: that mechanism, behavior and representation are three separate questions~\cite{PerezReche2026}, a distinction we fix below as the one between structure and mechanism; and that arity, the number of nodes an interaction couples, is not a universality label, so a collective phenomenon does not by itself classify an interaction order~\cite{VillegasMeloni2026}. We arrive at the same conclusions and add two formal results, namely an invariant of the map that no exact representation can alter (Section~\ref{sec:invariant}), and an explicit separation in bits between a hyperedge list and the edge list of its clique projection, under explicit structural codes (Section~\ref{sec:ex-group}).

There is nevertheless a gap between functional containment and scientific model choice. The class of all functions on a neighborhood automatically contains every restricted functional family on the same arguments. This is a set-inclusion statement, not a criterion for deciding which restriction is shortest, identifiable from finite data, stable under intervention, or aligned with an observed mechanism. For the same reason, fit cannot settle the matter. A larger class fits any given data set at least as well as a smaller one it contains, and this is why model comparison in statistics charges for complexity, through an information criterion, a Bayesian prior, or a description length~\cite{Akaike1974,rissanen1978modeling,barron1998minimum,Grunwald2007}. A familiar example from outside network science is planetary motion. A sufficiently flexible class of smooth curves contains the ellipses traced by Kepler's laws and fits the data at least as well in sample, yet nobody would take this as a reason to abandon Kepler's laws. Containment, in other words, does not argue for the containing class; if anything, it is what makes a complexity penalty necessary.

Our central point is that higher-order information is not eliminated merely because the same map is evaluated on a graph. When a hypergraph is projected onto a graph, information about how neighbors are grouped may be lost from the skeleton. An exact graph-based emulation must then recover that information in some other part of the model, normally the node rule. Schematically,
\begin{equation}
 \underbrace{\text{grouping in the structure}}_{\text{hypergraph}}
 \;\;
 \begin{array}{c}
 \xrightarrow{\;\text{exact emulation}\;}\\[-4pt]
 \xleftarrow{\;\text{reverse recoverability}\;}
 \end{array}
 \;\;
 \underbrace{\text{grouping in the rule}}_{\text{graph emulation}}.
 \label{eq:relocation}
\end{equation}
Without restrictions on the rule, moving the grouping into the rule costs at most a constant number of additional bits, whereas recovering it from the rule is not guaranteed; Section~\ref{sec:rc} makes this asymmetry precise. Under restrictions the same transfer may be cheap, expensive, or immaterial depending on the model class, so it is not by itself an argument for either representation.

Throughout, we use two terms in fixed senses that the debate tends to blur. \emph{Structure} means the grouping as it is recorded in the structural description; \emph{mechanism} means the joint dependence exhibited by the dynamics, in whichever part of the model it is carried, so that the interaction order of Section~\ref{sec:invariant} is a property of the mechanism. The distinction matters because each is routinely taken to license the other. The statement ``the model uses a graph'' is a claim about the structural description, ``the interaction is pairwise'' is a claim about the mechanism, and neither entails the other. A graph skeleton does not make the mechanism pairwise, since the emulation of Section~\ref{sec:emulation} leaves the joint dependence intact inside the rule; and a group-valued skeleton does not make the mechanism higher-order, since the diffusion of Section~\ref{sec:ex-diffusion} admits a hypergraph description while remaining additive in pairs. Synergistic contagion is a published instance of the first: the probability that a contact transmits along an edge is made to depend on the states of the target's other neighbors~\cite{GomezGardenes2016}, so the skeleton is a graph while the mechanism is not pairwise, and the continuous transition becomes explosive. In Section~\ref{sec:invariant} we show that the mechanism, unlike the structure, has a description that does not depend on the representation, and this is what keeps the two apart.

The argument is not that hypergraphs are always necessary. Phenomena such as abrupt transitions, multistability, and synchronization are prominent in higher-order models~\cite{ferraz2024contagion,Battiston2020}, but do not by themselves certify a higher-order structure. Indeed, each can arise in models with a purely dyadic structure: abrupt collective transitions occur in threshold and complex-contagion models, both in mean-field formulations~\cite{Granovetter1978} and on networks~\cite{Watts2002,CentolaMacy2007}; multistability arises in generalized contagion with dyadic exposures and memory of past contacts, whose critical-mass class has coexisting stable low and high states, again in a mean-field treatment~\cite{DoddsWatts2004}, all of them dyadic in structure though not in mechanism, since a threshold couples the exposures jointly; and explosive synchronization can emerge even in Kuramoto oscillators, whose coupling is pairwise in mechanism as well as in structure~\cite{GomezGardenes2011}. Surveys of explosive transitions in networks are given in Refs.~\cite{Boccaletti2016,DSouza2019}. The inference that fails is the one from a macroscopic signature to an interaction order. The reason, given in Section~\ref{sec:invariant}, is that the invariant is the interaction order of the map (for binary states, the family of essential interaction sets $\Ess_i(F)$), which a bifurcation diagram does not determine. Nor does every group-valued data set warrant treating groups as irreducible interaction units. At the same time, exact emulation by an unrestricted graph function does not establish pairwise additivity, compression, identifiability, or mechanistic equivalence. Under a comparison between specific, comparably regularized hypotheses, graphs may be preferred, hypergraphs may be preferred, or the available data may leave the choice unresolved.

We proceed as follows. Section~\ref{sec:setup} fixes notation and distinguishes four questions the debate tends to merge: what the structural skeleton records, what dynamics a formalism can reproduce, what a data set can identify, and what mechanism the model makes explicit. Section~\ref{sec:relocation} states what an unrestricted graph emulation does establish, closes the readings it does not support, and shows that the interaction order of a map is an invariant of the map, determined without reference to any structural language, so that relocation moves where the grouping is stored but cannot dissolve it. In Section~\ref{sec:rc} the resulting descriptions are priced with algorithmic information. An exact structure--rule description has an intrinsic floor, a translation can move information between its two parts without lowering it, and a comparison of structural languages carries information only under conditions we state. Section~\ref{sec:mdl} replaces the incomputable floor by a practical minimum-description-length (MDL) criterion for restricted and possibly noisy models, which makes the assumptions carrying the scientific content explicit. Section~\ref{sec:regimes} works through five examples, one of which (Section~\ref{sec:ex-group}) gives the cost, in bits, of describing the grouping through pairs; under explicit structural codes, the projected edge list of $M$ disjoint groups of size $k$ is longer than the hyperedge list by a factor approaching $k-1$. Section~\ref{sec:empirical} discusses the kinds of empirical evidence that can settle such comparisons.

\section{Setup}

\label{sec:setup}

\subsection{Notation and the structure--rule decomposition}
\label{sec:notation}
A model of a dynamical system contains at least two conceptually different objects: a structural description $S$ and a rule $D$ that acts on that structure. Their combination produces a dynamical map $F$,
\begin{equation}
 (S,D)\longmapsto F.
 \label{eq:model}
 \end{equation}
For instance, for a graph, $S$ may be an adjacency matrix; for a hypergraph it may be an incidence matrix; for a factor graph it may include both variable and factor nodes. The rule may be linear, additive, thresholded, stochastic, learned, or fully unrestricted, among other possibilities. It is useful to begin without privileging any structural language, so we fix notation for graphs and hypergraphs.

A \emph{graph} $G=(V,E)$ has $V=\{1,\ldots,N\}$ and $E\subseteq\{\{i,j\}:i,j\in V,\ i\neq j\}$. A \emph{directed graph} on $V$ has instead a set of ordered pairs $(j,i)$ with $i\neq j$, written $j\to i$. A \emph{hypergraph} $H=(V,\mathcal{E})$ has $\mathcal{E}\subseteq 2^{V}\setminus\{\emptyset\}$, and we write $|e|$ for the size of a hyperedge $e$. The \emph{clique projection} $\pi(H)$ is the graph with edge set $\{\{i,j\}:i\neq j,\ \exists\,e\in\mathcal{E},\ \{i,j\}\subseteq e\}$. The \emph{neighborhood} $\partial i=\{j:\{i,j\}\in E\}$ collects the nodes adjacent to $i$, and $x_{\partial i}=(x_j)_{j\in\partial i}$ denotes the corresponding vector of states, taken in a fixed order induced by the node labels. Node states take values in a set $\mathcal{X}$, finite unless stated otherwise, so that a single node state is $x_i\in\mathcal{X}$ and a configuration is $\bm{x}=(x_1,\ldots,x_N)\in\mathcal{X}^{N}$.

One further object is needed throughout. Given a map $F:\mathcal{X}^N\to\mathcal{X}^N$ with components $F_i$, its \emph{dependency graph} $G_F$ is the directed graph containing $j\to i$ whenever varying $x_j$ can change $F_i$ while the other variables are held fixed. We write $\partial_F i=\{j\neq i:\,j\to i\ \text{in}\ G_F\}$ for the in-neighborhood of $i$ and $|\partial_F i|$ for its size. It need not agree with $\partial i$: a declared graph may carry edges the dynamics does not use. It records which variables each update actually uses, independently of how the model happens to be written, and is the graph on which Section~\ref{sec:emulation} constructs its emulation.

\subsection{Four distinct questions}
\label{sec:four}
The graph--hypergraph debate often proceeds as though these objects could be ranked once and for all. In our view they cannot, because at least four different comparisons are in play and a representation can win one while losing another. The question becomes tractable once these comparisons are named separately.
\paragraph{Structural representation.} What information is explicit in $S$? A clique projection records whether two nodes co-occur in at least one hyperedge, but does not generally record which co-occurrences belong to the same hyperedge. A bipartite incidence graph can record the latter, although it does so by adding factor or group nodes. Both objects may be called graphs in a broad mathematical sense, but they encode different primitive entities.

\paragraph{Functional representability.} Can one model reproduce the same input--output map as another? With arbitrary functions, the answer is often yes. This is the legitimate content of unrestricted graph emulation. Functional representability is an existence claim; it need not preserve the decomposition of the rule, its regularity, or the length of its description.

Expressiveness and functional representability are the same notion at two levels of quantification. For a model class $\mathcal{M}=(\mathcal{L},\mathcal{C})$, a structural language together with a class of admissible rules, write
\[
 \mathcal{F}(\mathcal{M})=\{\,F:\ \exists\,S\in\mathcal{L},\,D\in\mathcal{C} \text{ with } D(S)=F\,\}
\]
for the maps it can realize. Expressiveness compares these sets (the claim that graphs are maximally expressive would imply $\mathcal{F}(\mathcal{M}_H)\subseteq\mathcal{F}(\mathcal{M}_G)$ for a graph class $\mathcal{M}_G$ with unrestricted rules and any hypergraph class $\mathcal{M}_H$), whereas functional representability is membership $F\in\mathcal{F}(\mathcal{M})$ for a particular map. We keep the pointwise term because the other three questions here are likewise asked of a given system and a given data set, and because ``expressiveness'' is the word whose loose use we are taking apart. The distinction has consequences. Proposition~\ref{prop:emulation} is class-level, while Sections~\ref{sec:ex-group} and~\ref{sec:ex-lift} turn on pointwise statements, and for this reason the proposition cannot settle them.

\paragraph{Statistical identifiability.} Can observations distinguish the candidate models? If one class strictly contains another, the larger class can attain training fit no worse than that of the smaller class, so fit alone cannot distinguish the two, nor justify choosing the smaller. Identifiability then depends on restrictions, priors, interventions, and the resolution at which interactions are observed. Comparison by description length, developed in Sections~\ref{sec:rc} and~\ref{sec:mdl}, belongs to this question: by charging for the complexity of each candidate, it can discriminate where fit alone cannot.

\paragraph{Mechanistic adequacy.} Does the representation expose the units, symmetries, and causal organization that the model purports to explain? Two descriptions may generate the same trajectory while assigning the regularity to different places. A hypergraph may state that a group is an interaction unit; an emulating graph may store the same grouping in an opaque node function. Equality of predictions does not make these hypotheses identical.

\vspace{0.5cm}
\noindent These notions need not rank models in the same way. A representation can be maximally expressive but statistically weak because it is too flexible. A mechanistically transparent model can be more expensive in bits if it includes independently measured structure. A short effective rule can predict well while being unsuitable for interventions outside the observed regime. The intended comparison therefore has to be stated first. Each example in Section~\ref{sec:regimes} closes by naming which of these questions decides the comparison; in several cases more than one is in play, and they need not agree.

\section{Emulation and higher-order information}
\label{sec:relocation}
The first subsection states the emulation argument: if node functions are unrestricted, a graph reproduces the node-level map of any hypergraph model. The remaining subsections separate that reproduction from a stronger claim often read into it, namely that the interaction was therefore never higher-order. Reproducing a map is one thing and accounting for how it arises is another, and most of the scientific content of the debate lies in the difference.

\subsection{Unrestricted graph emulation}
\label{sec:emulation}
Consider a finite-state hypergraph dynamical system on $N$ labeled nodes, with node-level map $F$ and dependency graph $G_F$ as defined in Section~\ref{sec:notation}.
\begin{proposition}[Unrestricted graph emulation]
\label{prop:emulation}
If arbitrary node functions are admissible on $G_F$, then the graph model can reproduce the same node-level map $F$ by assigning $F_i$ itself as the rule of node $i$.
\end{proposition}

\begin{proof}
By definition of $G_F$, $F_i$ depends only on $x_i$ and the variables in the in-neighborhood of $i$ in $G_F$. Hence $F_i$ is an admissible multivariate node function on that graph. Assigning it to each node reproduces $F$ exactly.
\end{proof}

The same observation holds for continuous variables when the admissible graph rule class contains the induced functions with their required regularity. It also extends beyond hypergraphs: any finite packaging of the arguments of $F_i$ can be forgotten once the complete function $F_i$ is retained. Proposition~\ref{prop:emulation} has two important consequences. First, macroscopic behavior is not a certificate of structural order. If an abrupt transition observed in a hypergraph model can be produced by another rule on a graph, the transition alone cannot identify the microscopic representation. Second, even complete node-state trajectories cannot distinguish a hypergraph model from an unrestricted graph class that contains its exact emulation. Fit alone is therefore insufficient.

The proposition does not compare unrestricted graph rules with an equally unrestricted higher-order language. The strict containment discussed in Ref.~\cite{Peixoto2026} is valid for the particular comparison between arbitrary neighborhood functions and the hypergraph rule family defined there, in which contributions associated with incident hyperedges are aggregated under specified reciprocity and symmetry constraints, and from which hypergraphs whose hyperedges all have size two, or whose rules ignore the hyperedges, are explicitly excluded. This qualification matters, because restricting one class while leaving the other unrestricted establishes a hierarchy between those two model classes and not a universal ordering of structural languages.

Taken on its own, Proposition~\ref{prop:emulation} invites a stronger reading than it supports. It does not imply that the graph skeleton contains the hypergraph, that the rule is pairwise additive, or that regularity and description length are preserved. The next three subsections close those readings, together with one further inference that travels with them. Section~\ref{sec:groups} shows that the projection discards the grouping, which the rule must then carry; Section~\ref{sec:recoding} shows that recoding the arguments makes a function look unary without making it additive, well behaved, or shorter; and Section~\ref{sec:compressibility} closes the further inference, that a compressible rule must decompose into low-order parts, by showing that even a very short rule can be irreducibly high-order. Section~\ref{sec:invariant} then collects these observations into a positive statement. The interaction order is fixed by the map alone, and for binary states the joint dependence has a canonical form that no exact representation can alter.

\subsection{A projected neighborhood does not determine its groups}
\label{sec:groups}
Let $H$ be a hypergraph and $\pi(H)$ its clique projection. The variables available to node $i$ in the projected graph form the union of the other members of its incident hyperedges. The map from a family of groups to this union is many-to-one.

\begin{lemma}[Grouping ambiguity]
\label{lem:union}
Distinct collections of subsets can have the same union and the same clique projection.
\end{lemma}

\begin{proof}
On four nodes, consider
\begin{equation}
 H_1=\bigl\{\{1,2,3,4\}\bigr\}
 \quad\text{and}\quad
 H_2=\bigl\{\{1,2,3\},\{1,2,4\},\{1,3,4\},\{2,3,4\}\bigr\}.
 \label{eq:H12}
\end{equation}
Both project to the complete graph on the four nodes, $\pi(H_1)=\pi(H_2)=K_4$, although their hyperedge families are different.
\end{proof}

Now let node $i$ activate when all other members of at least one incident hyperedge are active. For node $1$ the two structures in Eq.~\eqref{eq:H12} give
\begin{align}
 F^{(1)}_1(x_2,x_3,x_4)&=x_2x_3x_4,
 \label{eq:group-functions-1}\\
 F^{(2)}_1(x_2,x_3,x_4)&=(x_2x_3)\lor(x_2x_4)\lor(x_3x_4).
 \label{eq:group-functions}
\end{align}
At $(x_2,x_3,x_4)=(1,1,0)$ the first output is $0$ and the second is $1$. The projected graph and its node states are identical; the difference lies in the grouping. This does not contradict Proposition~\ref{prop:emulation}: for a fixed and known $H_1$ or $H_2$ one may define the corresponding function directly on the neighborhood of $K_4$, but the function is then indexed by the hypergraph that has been removed from the skeleton. The emulation succeeds because the rule carries the missing information.

More generally, suppose a graph node rule is expressed as

\begin{equation}
 f_i(x_i,x_{\partial i})=g_i(x_i)+
 \sum_{\ell=1}^{m_i}h_{i\ell}(x_i,x_{Y_{i\ell}}),
 \qquad Y_{i\ell}\subseteq \partial i.
 \label{eq:subset-rule}
\end{equation}
This form also appears in Ref.~\cite{Peixoto2026}. The graph $G$ determines $\partial i$ but not the selected family $\mathcal{Y}_i=\{Y_{i1},\ldots,Y_{im_i}\}$. Consequently, the complete model is
\begin{equation}
 \bigl(G,\{\mathcal{Y}_i\}_{i=1}^N,
 \{g_i\}_{i=1}^N,\{h_{i\ell}\}_{i,\ell}\bigr),
 \label{eq:complete-model}
\end{equation}
not the graph alone. Whether $\mathcal{Y}_i$ is classified as structure or as part of the rule is a matter of representation, but it has to be specified somewhere.

\begin{proposition}[The emulating model carries an incidence structure]
\label{prop:incidence}
Let the rule have the form~\eqref{eq:subset-rule}. Then:~\emph{(i)} the data $\bigl(\{\mathcal{Y}_i\},\{g_i\},\{h_{i\ell}\}\bigr)$ determine $F$, and with it the dependency graph $G_F$, so that $G$ may be omitted from Eq.~\eqref{eq:complete-model};~\emph{(ii)} for every $i$,
\begin{equation}
 \partial_F i\subseteq\bigcup_{\ell=1}^{m_i}Y_{i\ell};
 \label{eq:union-recovers}
\end{equation}
and~\emph{(iii)} the assignment $i\mapsto\{\,\{i\}\cup Y_{i\ell}\,\}_{\ell=1}^{m_i}$ is the incidence structure of a hypergraph on $V$.
\end{proposition}
\begin{proof}
The right-hand side of Eq.~\eqref{eq:subset-rule} refers to the graph nowhere except through the side condition $Y_{i\ell}\subseteq\partial i$, so every graph whose neighborhoods contain the selected subsets yields the same $f_i$; the remaining data therefore determine $f_i$ for every $i$, hence $F$ and $G_F$, which is~\emph{(i)}. For \emph{(ii)}, if some $j\in\partial_F i$ occurred in no $Y_{i\ell}$ then $f_i$ would not depend on $x_j$ and $j\to i$ would be absent from $G_F$. Part~\emph{(iii)} restates the definition of an incidence structure: each node is assigned a family of subsets of $V$ containing it.
\end{proof}

The hypergraph in~\emph{(iii)} is directed in general, with node $i$ as the head of each set $\{i\}\cup Y_{i\ell}$. Collecting these sets over all nodes gives an ordinary hypergraph, but $j\in Y_{i\ell}$ does not oblige the same set to appear in the rule of $j$, so the assignment also records which rule each set enters. The hypergraph is undirected when the selections are mutual, as for the shared rule of Section~\ref{sec:ex-group}. Proposition~\ref{prop:incidence} makes clearer what the emulation achieves. Neither graph is a primitive of the model. The declared one does no work that the selected subsets do not already do, and the dependency graph is derived from the same data. The subsets, on the other hand, cannot be dropped. What looks like a graph model is then a higher-order model that happens to admit a graph as a derived display, with the grouping renamed but still present.

Ref.~\cite{Peixoto2026} raises an objection that has to be addressed first: every $f_i$ admits the trivial decomposition $m_i=1$, $Y_{i1}=\partial_F i$, whose induced hyperedges are the closed in-neighborhoods of $G_F$. Read that way the proposition is empty, since the family $\{\mathcal{Y}_i\}$ is then computable from $G_F$ and names nothing new. The proposition has content because the decomposition~\eqref{eq:subset-rule} is not arbitrary. A rule of interaction order below $|\partial_F i|$ admits decompositions with $|Y_{i\ell}|<|\partial_F i|$, and the selected subsets then carry information the adjacency does not. Proposition~\ref{prop:incidence} is therefore a statement about the parameterization chosen, and it inherits whatever arbitrariness that choice has. Section~\ref{sec:invariant} removes the arbitrariness by replacing $\{\mathcal{Y}_i\}$ with a family determined by $F$ alone.

\subsection{Scalar recoding does not decompose interaction}
\label{sec:recoding}
Another route to an exact emulation is to encode a vector as a scalar. Let $\bm{x}\in\{0,\ldots,A-1\}^k$ and consider the positional code
\begin{equation}
 \psi:\{0,\ldots,A-1\}^k\longrightarrow\{0,\ldots,A^k-1\},
 \qquad
 \psi(\bm{x})=\sum_{r=1}^k A^{r-1}x_r,
 \label{eq:positional}
\end{equation}
which is a bijection onto its stated codomain. For any $f:\{0,\ldots,A-1\}^k\to\mathcal{X}$, the state set of Section~\ref{sec:notation}, we may therefore define $h:\{0,\ldots,A^k-1\}\to\mathcal{X}$ by $h(z)=f(\psi^{-1}(z))$, so that $f(\bm{x})=h(\psi(\bm{x}))$. Since $\psi$ is itself a sum of univariate terms, this is the single-index form $f(\bm{x})=h\bigl(\sum_r a_r(x_r)\bigr)$ with $a_r(x_r)=A^{r-1}x_r$, which is the sense in which any multivariate rule can be written as an outer function of one contribution per argument~\cite{Peixoto2026}. The construction is exact, and it makes the rule look simpler, since a function of $k$ arguments has become a function of one. The simplification, however, concerns only the arity. The domain of $h$ still has $A^k$ elements, so for a generic $f$ it carries the same table of values in the same number of bits. Collapsing the arguments renames the inputs without compressing the rule.

The construction uses no property of $\mathcal{X}$ beyond its being a set, and no property of $f$ beyond its being a function. An argument that is insensitive to the values a rule takes cannot be evidence about how those values depend jointly on the inputs.

After the recoding $h$ has a single argument, and a function of one argument shows no interaction among its arguments. It is tempting to read that as an absence of interaction in the system, but interaction is a property relative to the original variables, not to a coordinate that bundles them. XOR makes the point concrete. Consider $f(x_1,x_2)=x_1\oplus x_2$. It can be written as $h(x_1+2x_2)$, with $h(0)=h(3)=0$ and $h(1)=h(2)=1$. However, there are no real-valued functions $a,b$ and constant $c$ satisfying
\begin{equation}
 x_1\oplus x_2=c+a(x_1)+b(x_2)
 \label{eq:xor-additive}
\end{equation}
for every binary input. Every additively separable function obeys
\begin{equation}
 f(0,0)+f(1,1)=f(1,0)+f(0,1),
\end{equation}
whereas XOR would require $0=2$. The joint dependence has been hidden in $h$ without being decomposed. The relevant distinction is between $h\bigl(\sum_r a_r(x_r)\bigr)$ and $\sum_r h_r(x_r)$: the recoding delivers the first, interaction order is a statement about the second, and XOR has the first form but not the second.

For continuous variables the recoding meets an obstacle with no finite counterpart. No bijection $\mathbb{R}^k\to\mathbb{R}$ is continuous for $k>1$~\cite{Munkres}\footnote{If $f:\mathbb{R}^k\to\mathbb{R}$ were a continuous bijection, then for any point $p$ the set $\mathbb{R}^k\setminus\{p\}$ would be connected for $k>1$ while its image $\mathbb{R}\setminus\{f(p)\}$ is not.}, so a model class that requires regular rules (as any well-posed ODE does, local Lipschitz continuity being the standard sufficient condition for existence and uniqueness) excludes $\psi$ outright. The Kolmogorov--Arnold theorem is not a counterexample~\cite{Kolmogorov1957,Arnold1957}: it uses $2k+1$ scalar channels, not one bijective coordinate.

\subsection{Compressibility and interaction order are independent}
\label{sec:compressibility}
Description length and interaction order are easily conflated, because a rule coupling $k$ variables has a truth table of $2^k$ rows, and a large table suggests a long description. The two are in fact independent. Compressibility concerns the number of bits needed to specify a rule, whereas interaction order concerns how many variables must enter jointly in an additive decomposition over the original variables. A rule can be very short and still couple every variable at once.

The parity function provides a simple illustration:
\begin{equation} p_k(x_1,\ldots,x_k)=x_1\oplus\cdots\oplus x_k
 =\frac{1}{2}\left[1-\prod_{r=1}^k(1-2x_r)\right],  \label{eq:parity}
\end{equation}
where $x_r\in\{0,1\}$ and $\oplus$ denotes addition modulo two. The output is one precisely when an odd number of inputs are one. Although the truth table contains $2^k$ rows, the rule has the short uniform description ``return the parity of the $k$ inputs.'' Formally, there is a program of constant length that computes $p_k$ given $k$, so $\K(p_k)\leq\K(k)+O(1)=O(\log k)$ for the Kolmogorov complexity $\K$ defined in Section~\ref{sec:rc}. The exponential size of the truth table should not be confused with the algorithmic description length of the rule.

Despite this short description, parity contains an irreducible joint dependence on all $k$ inputs. For example, when $k=3$,
\begin{equation}
\begin{split}
 p_3(x_1,x_2,x_3)
 ={}&x_1+x_2+x_3
 -2(x_1x_2+x_1x_3+x_2x_3)\\
 &+4x_1x_2x_3.
\end{split}
\label{eq:parity-three}
\end{equation}
The final term involves the three variables simultaneously. More generally, expanding the product in Eq.~\eqref{eq:parity} produces a nonzero term proportional to $x_1x_2\cdots x_k$, with coefficient $(-1)^{k+1}2^{k-1}$.

The mixed-difference operator makes the irreducibility precise. For a function on the Boolean cube, define
\begin{equation}
 \bigl(\Delta_i f\bigr)(\bm{x}_{-i})
 =
 f(x_1,\ldots,x_{i-1},1,x_{i+1},\ldots,x_k)
 -
 f(x_1,\ldots,x_{i-1},0,x_{i+1},\ldots,x_k).
 \label{eq:mixed-difference}
\end{equation}
If a term does not depend on $x_i$, then $\Delta_i$ annihilates it. Suppose that a function could be decomposed as
\begin{equation}
 f(x_1,\ldots,x_k)=\sum_{A\subsetneq\{1,\ldots,k\}} f_A(\bm{x}_A),
 \label{eq:lower-order-decomposition}
\end{equation}
so that every component omits at least one of the $k$ variables. Applying the full mixed difference $\Delta_1\cdots\Delta_k$ would annihilate every term in the sum and therefore give zero. Parity does not satisfy this condition:
\begin{equation}
 \Delta_1\cdots\Delta_k p_k
 =
 (-1)^{k+1}2^{k-1}\neq 0.
 \label{eq:parity-mixed-difference}
\end{equation}
Consequently, parity cannot be expressed as a sum of lower-order components, even though the algorithm specifying it is very short.

This argument does not imply that a graph cannot reproduce parity. If all $k$ inputs belong to the neighborhood of a node, an unrestricted graph-based node function can certainly evaluate Eq.~\eqref{eq:parity}. However, the resulting node function still depends jointly on all $k$ neighbors. The graph representation has placed the high-order dependence inside the dynamical rule; it has not converted it into pairwise-additive interactions.

Parity is therefore a counterexample to the implication ``compressible rule implies low interaction order.'' The converse implication fails just as easily. A rule $f_i(x_{\partial i})=\sum_{j\in\partial i}w_{ij}x_j$ whose weights are drawn at random and stored to $b$ bits is additively separable, so no two neighbors enter jointly and its interaction order is one, yet its description length is $\Theta(b|\partial i|)$, incompressible for all but a vanishing fraction of weight assignments. Long descriptions of low-order mechanisms and short descriptions of high-order ones both exist.

This bears on an argument made for the graph side. Reference~\cite{Peixoto2026} notes that specifying a function of $k$ arguments over an alphabet of size $A$ requires $A^{k}\log_2 A$ bits in general, which is beyond reach already for modest $k$, and concludes that the functions available for modeling are compressible and hence should not be treated as irreducible monoliths. We agree with the premise and with the first inference, but not with the last step, since parity is as compressible as a rule can be and is still irreducibly of order $k$. Compressibility makes modeling possible; it does not by itself license a decomposition into lower-order parts.

In short, exact emulation does not make interaction order, regularity or grouping disappear; it shows only that they can be encoded elsewhere. The next subsection turns this into an exact statement.

\subsection{Interaction order is an invariant of the dynamics}
\label{sec:invariant}
The mixed differences of Eq.~\eqref{eq:mixed-difference} do more than settle the case of parity. We now apply them to the components $F_i$ of the full map, so that subsets are indexed by nodes, $A\subseteq V$, rather than by the argument positions $1,\ldots,k$ of Section~\ref{sec:compressibility}. When node states are binary, $\mathcal{X}=\B$, so that $\bm{x}$ ranges over the Boolean cube $\B^{N}$, every component of $F$ has a unique multilinear expansion~\cite{ODonnell2014},
\begin{equation}
 F_i(\bm{x})=\sum_{A\subseteq V}c_{i,A}\prod_{j\in A}x_j,
 \qquad
 c_{i,A}=\bigl(\Delta_A F_i\bigr)\big|_{\bm{x}=\bm{0}},
 \qquad
 \Delta_A=\prod_{j\in A}\Delta_j,
 \label{eq:multilinear}
\end{equation}
the coefficients being obtained by M\"obius inversion over the Boolean cube. Explicitly, this is inclusion and exclusion over the subsets of $A$,
\[
 c_{i,A}=\sum_{B\subseteq A}(-1)^{|A|-|B|}F_i(\ind_B),
\]
where $\ind_B$ denotes the configuration with the nodes of $B$ set to one and the rest to zero; applying $\Delta_j$ for every $j\in A$ and evaluating at $\bm{x}=\bm{0}$ produces the same alternating sum. For a node whose update depends on two variables, with $F_i=x_1\oplus x_2$ as in Section~\ref{sec:recoding}, this gives $c_{i,\{1,2\}}=F_i(1,1)-F_i(1,0)-F_i(0,1)+F_i(0,0)=-2$ and $c_{i,\{1\}}=c_{i,\{2\}}=1$, so that $F_i=x_1+x_2-2x_1x_2$. Call
\begin{equation}
 \Ess_i(F)=\bigl\{A\subseteq V:\ c_{i,A}\neq0,\ |A|\geq2\bigr\},
 \qquad
 k_i(F)=\max\bigl\{|A|:\ c_{i,A}\neq0\bigr\}
 \label{eq:essential}
\end{equation}
the \emph{essential interaction sets} and the \emph{interaction order} of the update of node $i$.

\begin{proposition}[Order is the degree of the expansion]
\label{prop:order}
$F_i$ is a sum of functions each depending on at most $d$ of the variables if and only if $c_{i,A}=0$ whenever $|A|>d$. Hence $k_i(F)$ is the least such $d$.
\end{proposition}
\begin{proof}
Sufficiency is immediate, since each monomial of the expansion involves at most $d$ variables. For necessity, let $F_i=\sum_B f_B$ with $|B|\leq d$ and take $|A|>d$. Every $f_B$ omits some $j\in A$, so $\Delta_j f_B=0$; the operators $\Delta_j$ commute, hence $\Delta_A F_i=0$ and $c_{i,A}=0$. Equations~\eqref{eq:lower-order-decomposition}--\eqref{eq:parity-mixed-difference} are the case $d=k-1$ applied to parity.
\end{proof}

The decomposition condition in Proposition~\ref{prop:order} makes no reference to the multilinear expansion, so it extends to any state set $\mathcal{X}\subseteq\mathbb{R}$, finite, as in Section~\ref{sec:recoding}, or continuous. For such $\mathcal{X}$ we define $k_i(F)$ as the least $d$ for which $F_i=\sum_{\ell}f_{\ell}(\bm{x}_{B_{\ell}})$ with $B_{\ell}\subseteq V$ and $|B_{\ell}|\leq d$ for every $\ell$; this least value exists because $f_1=F_i$ makes $d=N$ admissible, and by Proposition~\ref{prop:order} it agrees with Eq.~\eqref{eq:essential} when $\mathcal{X}=\B$. So defined, $k_i(F)$ refers only to the values of $F_i$ and to no structural language, and the invariance below holds for it as it stands. What the binary case adds is computability and resolution: the expansion~\eqref{eq:multilinear} evaluates the order and refines it into the family $\Ess_i(F)$, which is defined only in that case.

\begin{proposition}[Invariance under exact representation]
\label{prop:invariance}
$\Ess_i(F)$ and $k_i(F)$ are determined by $F$ alone. Consequently any two exact representations of the same map in the same node variables (on a graph, a hypergraph, a factor graph, or a multilayer graph, with any admissible rule class) have the same essential interaction sets and the same interaction order; and $j\to i$ lies in the dependency graph $G_F$ exactly when $j$ belongs to some $A$ with $c_{i,A}\neq0$.
\end{proposition}
\begin{proof}
By Eq.~\eqref{eq:multilinear} each $c_{i,A}$ is a fixed linear combination of values of $F_i$, so it does not depend on how $F$ was written. For the second claim, $\Delta_jF_i=\sum_{A\ni j}c_{i,A}\prod_{r\in A\setminus\{j\}}x_r$ is multilinear and, by uniqueness of the expansion, vanishes identically only if every $c_{i,A}$ with $A\ni j$ vanishes.
\end{proof}

Here the two senses fixed in Section~\ref{sec:introduction} come apart. The essential interaction sets are a property of the mechanism and survive every change of structural language, while the recorded grouping is a property of $S$ and does not. Three consequences follow. The emulation of Proposition~\ref{prop:emulation} preserves $\Ess_i(F)$ trivially, since it assigns $F_i$ itself, so it cannot be evidence that the dependence was never higher-order; the arbitrariness of Proposition~\ref{prop:incidence} is removed, $\Ess_i(F)$ being canonical where $\mathcal{Y}_i$ was a choice; and a rule family capped below $k_i(F)$ cannot represent $F_i$ at all, which is the misspecification priced in Section~\ref{sec:ex-group}.

The closest precedent is the effective order of Ref.~\cite{Neuhauser2024}, the least order at which the dynamics of a hypergraph model of declared form can still be written exactly. Within that format it plays the role of $k_i(F)$, being what any exact representation must carry; Proposition~\ref{prop:invariance} removes the dependence on the format, since $k_i(F)$ is fixed by $F$ for every structural language and rule class. The order that Ref.~\cite{Neuhauser2024} learns from data is instead selected by a score that penalizes complexity, and can be lower; the gap between the two is a predictive loss, which Section~\ref{sec:mdl} prices in bits.

The invariant carries more information than the projection: two maps with the same dependency graph, here $K_4$, can have different essential interaction sets. It does not, however, coincide with the hyperedge family. For the two structures of Eq.~\eqref{eq:H12},
\begin{equation}
 \Ess_1\bigl(F^{(1)}\bigr)=\bigl\{\{2,3,4\}\bigr\},
 \qquad
 \Ess_1\bigl(F^{(2)}\bigr)=\bigl\{\{2,3\},\{2,4\},\{3,4\},\{2,3,4\}\bigr\},
 \label{eq:essential-H12}
\end{equation}
using $F^{(2)}_1=x_2x_3+x_2x_4+x_3x_4-2x_2x_3x_4$. The two mechanisms share the dependency graph $K_4$ and the interaction order $3$ yet have different essential interaction sets, so Eq.~\eqref{eq:essential}, applied to the map $F$ itself, distinguishes what Lemma~\ref{lem:union} shows the projection to conflate. Under the correspondence $A\mapsto\{i\}\cup A$, however, the set $\{2,3,4\}\in\Ess_1(F^{(2)})$ names the hyperedge $\{1,2,3,4\}$, which is in $H_1$, not $H_2$. Each hyperedge of $H_2$ incident to node~$1$ contributes only a quadratic term, while the cubic term $-2x_2x_3x_4$ comes from no hyperedge at all: the disjunction produces it by inclusion and exclusion over the three groups. The order of a mechanism can therefore exceed $|e|-1$ for every incident group, and $\Ess_i(F)$ reflects the structure together with the rule's aggregation over it. The joint dependence is representation-independent, but the way it is assigned to structural elements is not, and for this reason the comparison of structures in the next section has to be conducted in bits.

Two qualifications delimit these statements. The expansion~\eqref{eq:multilinear} is taken relative to the reference configuration $\bm{x}=\bm{0}$, the state in which every other node is inactive, so $c_{i,A}$ is the joint effect of the nodes in $A$ against a quiescent background. Relabeling the states of a node, $x_j\mapsto1-x_j$, changes the lower-order coefficients and can change $\Ess_i(F)$ (in the fully relabeled coding $F^{(1)}_1$ and $F^{(2)}_1$ have the same family), whereas $k_i(F)$, the inclusion-maximal sets of $\Ess_i(F)$, and $G_F$ do not depend on the reference. For continuous states and sufficiently smooth maps on a product of intervals, $k_i(F)$ is likewise the least $d$ for which every mixed partial derivative $\partial^{|A|}F_i/\prod_{j\in A}\partial x_j$ with $|A|=d+1$ vanishes identically. The sets $A$ whose mixed partial derivative does not vanish identically form a reference-free family, closed under taking subsets and therefore determined by its inclusion-maximal members; in the binary case, with $\Delta_A$ in place of the derivative, its members of size at least two are the sets of that size contained in some member of $\Ess_i(F)$. A family that also resolves the lower-order terms, as $\Ess_i(F)$ does, would again require a reference point. For $F_1=\sin(x_2+x_3+x_4)$ every subset of $\{2,3,4\}$ belongs to the family, yet among the mixed derivatives of order two or more only the third-order one is nonzero at $\bm{x}=\bm{0}$, and only the second-order ones at a point with $x_2+x_3+x_4=\pi/2$.

\section{Representational complexity}
\label{sec:rc}
We now formalize the cost of a structure--rule decomposition. The general framework is developed in Ref.~\cite{Rommens2026RC}; we recall only what is needed here and then specialize to the pairwise--higher-order comparison. All objects in this section are finite strings, and $\K(x)$ denotes the prefix Kolmogorov complexity of $x$, the length of a shortest program that outputs $x$ on a fixed universal prefix-free machine~\cite{LiVitanyi2019}. Conditional complexity is written $\K(x\mid y^*)$, where $y^*$ is a shortest program for $y$, and the prefix chain rule reads $\K(x,y)\aeq\K(x)+\K(y\mid x^*)$. Relations marked with $+$ hold up to an additive constant independent of system size.

\subsection{The floor and the two-part split}
\label{sec:rc-floor}
A representation of a finite dynamical map $F$ is a pair $(S,D)$ such that a fixed evaluator applied to $D$ and $S$ outputs $F$. Its representational complexity is the excess description length above the map itself,
\begin{equation}
 \RC(S,D;F)=\K(S,D)-\K(F).
 \label{eq:RC}
\end{equation}

\begin{proposition}[Description-length floor]
\label{prop:floor}
Every exact representation satisfies $\RC(S,D;F)\ageq0$. Moreover, the structure-free representation that stores a shortest program for $F$ in $D$ attains $\RC\aeq0$.
\end{proposition}

\begin{proof}
Since a fixed evaluator computes $F$ from $(S,D)$, $\K(F)\aleq\K(S,D)$. Conversely, take $S$ to be the empty string and let $D$ output $F$; this requires $\K(F)+O(1)$ bits.
\end{proof}

Proposition~\ref{prop:floor} is meant as a benchmark and not as a recommendation to discard structure, since a structure-free program generally provides little causal or mechanistic organization. All it says is that structure cannot beat the intrinsic algorithmic description of the map. The chain rule separates the cost into
\begin{equation}
 \K(S,D)\aeq\K(S)+\K(D\mid S^*).
 \label{eq:two-part}
\end{equation}
This decomposition makes information relocation visible. A richer skeleton may increase $\K(S)$ while reducing the conditional cost of the rule; a sparser skeleton may do the reverse. The floor itself plays no part in any comparison. Because $\K(F)$ is common to every representation of the same map, it cancels whenever two are compared, and what remains is a difference of two-part codes rather than of absolute complexities.

\subsection{The unrestricted comparison}
\label{sec:projection-cost}
Suppose $G=\pi(H)$ for a fixed computable projection $\pi$. Then $\K(G\mid H^*)=O(1)$, and the chain rule, applied in both orders, gives
\begin{equation}
 \K(G)+\K(H\mid G^*)\aeq\K(H).
 \label{eq:conditional-H}
\end{equation}
The relevant construction is not merely that $G$ is computable from $H$, but that there is a fixed computable compiler $C$ acting on complete model descriptions,\footnote{The compilers used here are fixed, uniform computable functions between model descriptions that preserve the represented dynamics. A one-way compiler supplies such a map in one direction; a two-way translation supplies fixed computable maps in both. They need not be literal inverses, since distinct descriptions may encode the same dynamics. This is the uniformity convention of Ref.~\cite{Rommens2026RC}: translators are fixed finite objects, independent of system size, so the constants they contribute are bounded uniformly.}
\begin{equation}
 C(H,D_H)=(G,D_G),
 \label{eq:compiler}
\end{equation}
which constructs $G=\pi(H)$ and emits a graph rule $D_G$ that evaluates the original hypergraph rule while retaining whatever grouping information it needs. Producing $(G,D_H)$ would not suffice, since $D_H$ need not be an admissible graph rule. From Eq.~\eqref{eq:compiler},
\begin{equation}
 \K(G,D_G)\aleq\K(H,D_H).
 \label{eq:no-double-count}
\end{equation}

The non-invertibility of the projection therefore does not by itself show that the hypergraph has a shorter optimal description. The one-way compiler establishes only that a graph emulation need not be longer than the source hypergraph description by more than a constant. Equality of the two optimal lengths requires more, for instance a fixed reverse compiler recovering an admissible $(H,D_H)$ from $(G,D_G)$; without such reverse recoverability, Eq.~\eqref{eq:no-double-count} leaves the possibility that the graph description is strictly shorter, and strict separations of this kind belong to the restricted setting of Section~\ref{sec:restricted}. In a two-part view, an efficient graph code charges only for the residual choice of $H$ within the projection preimage
\begin{equation}
 \pi^{-1}(G)=\{H:\pi(H)=G\},
 \label{eq:preimage}
\end{equation}
rather than paying independently for a full graph and a full hyperedge list. A one-way translation can thus avoid any additional penalty while changing where the information is stored, and an equivalence of descriptions needs an appropriate two-way translation. A preference between scientific models needs more than either, namely explicit restrictions on admissible structures, rules, translators, regularity classes, or code families.

Equation~\eqref{eq:no-double-count} has a probabilistic counterpart in the discussion of parsimony in Ref.~\cite{Peixoto2026}, where marginalizing a hypergraph model over the hypergraphs and rules that induce each graph model yields a graph model whose description length is no larger, and strictly smaller when several hypergraph parameterizations induce the same graph model. The prior over graph rules obtained in this way is supported only on rules induced by hypergraph rules, and evaluating it requires a sum over hypergraphs, so the grouping has been moved into the prior over rules rather than eliminated, as in Eq.~\eqref{eq:relocation}.

\subsection{Restricted classes and the conditions for comparison}
\label{sec:restricted}
Let $\mathcal{M}=(\mathcal{L},\mathcal{C})$ consist of a structural language $\mathcal{L}$ and a class of rules $\mathcal{C}$, and let
\begin{equation}
 \mathcal{K}_{\mathcal{M}}(F)=
 \min_{\substack{S\in\mathcal{L},\,D\in\mathcal{C}\\D(S)=F}}
 \left[\K(S)+\K(D\mid S^*)\right]
 \label{eq:restricted}
\end{equation}
be the ideal restricted description length of $F$, defined when the feasible set is non-empty. For graph and hypergraph classes $\mathcal{M}_G$ and $\mathcal{M}_H$ that both admit a representation of $F$, define
\begin{equation}
 \Delta\RC=
 \mathcal{K}_{\mathcal{M}_G}(F)-
 \mathcal{K}_{\mathcal{M}_H}(F).
 \label{eq:deltaRC}
\end{equation}

Since $\K(F)$ is common to both classes and cancels (Section~\ref{sec:rc-floor}), $\Delta\RC$ is the difference of the corresponding restricted representational complexities. A positive value favors the hypergraph class, a negative value favors the graph class, and equality up to constants defines an equivalence regime.

Equation~\eqref{eq:deltaRC} carries information only if the two classes are not nested under a fixed translation. If $\mathcal{L}_G$ is defined syntactically, as the hypergraphs all of whose edges have size two, then every admissible graph representation is also an admissible hypergraph representation at $O(1)$ translation cost, the minimum defining $\mathcal{K}_{\mathcal{M}_H}(F)$ ranges over a superset of the one defining $\mathcal{K}_{\mathcal{M}_G}(F)$, and $\Delta\RC\ageq 0$ identically, provided, as we assume here, that the admissible rules are nested along with the structures. The graph-preferred regime is then empty and the comparison vacuous.

A nontrivial sign therefore requires classes that are not nested, and of the two routes to that condition only one is available at this level. One route is to restrict the rule classes, and it works: if admissible graph and hypergraph rules each exclude functions the other permits, neither minimum in Eq.~\eqref{eq:restricted} ranges over a superset of the other. The other route, often invoked, is to individuate a structural language by the entities it posits as primitive, and it does not work here. A structure enters Eq.~\eqref{eq:restricted} only through $\K(S)$, which is a property of the string and not of the vocabulary used to read it, so whenever a hyperedge family and a dyadic skeleton are computable from each other their costs agree up to a constant however the two languages are individuated, exactly as in Section~\ref{sec:projection-cost}.

Separations obtained by individuating languages are real, but they are separations between codes rather than between complexities, and they appear only once Eq.~\eqref{eq:MDL} replaces $\K$ by a declared code family; Proposition~\ref{prop:projection-cost} exhibits one. Which rule classes are admissible is a modeling decision and should be stated; what is not available is a comparison in which one language is a relabeling of the other.

Subject to this condition there is no universal sign. If fixed, system-size-independent programs translate admissible graph representations into admissible hypergraph representations and back, the two minima differ by at most $O(1)$. If a translation produces a rule outside the target class --- a discontinuous lookup table, a nonlocal rule, or a function whose parameter count grows exponentially with degree --- the equivalence fails. It is at this point that modeling assumptions start to matter.

Finally, two comparison problems should be kept separate. In a \emph{fixed-skeleton} problem, $G$ may be an observed physical network and the question is whether additional group structure improves the explanation of the dynamics. In an \emph{optimized-structure} problem, both $G$ and $H$ are latent and selected jointly with their rules. A graph optimized for the map may omit edges that are redundant in an observed projection, while a fixed empirical graph cannot. Mixing these problems creates artificial description-length advantages for whichever side is allowed to optimize.

\section{An operational model-selection criterion}
\label{sec:mdl}
Kolmogorov complexity supplies the conceptual floor but is not computable. In data analysis it should be replaced by explicit codes and a likelihood. Minimum-description-length reconstruction of ordinary networks from dynamics is by now a mature methodology~\cite{Peixoto2019reconstruction,Peixoto2025mdl}, and the same discipline of charging explicitly for structure rather than selecting it by descriptive heuristics applies equally to the choice between pairwise and higher-order representations~\cite{Lambiotte2019,PeelPeixotoDeDomenico2022}; an information-theoretic criterion for redundancy among the hyperedge orders of a given hypergraph is already available on the structural side~\cite{Kirkley2025}. Let $\mathcal{D}$ be a set of observed transitions, trajectories, or event records. For a model class $\mathcal{M}=(\mathcal{L},\mathcal{C})$, define
\begin{equation}
 \Len_{\mathcal{M}}(\mathcal{D})=
 \min_{\substack{S\in\mathcal{L},\,D\in\mathcal{C}}}
 \left[
 \Len_{\mathcal{L}}(S)
 +\Len_{\mathcal{C}}(D\mid S)
 -\log_2 p(\mathcal{D}\mid S,D)
 \right].
 \label{eq:MDL}
\end{equation}
The first term encodes the structure, the second the rule conditional on it, and the third the data left unexplained by the model. Since $p(\mathcal{D}\mid S,D)\leq1$, the last term is non-negative and grows as the fit deteriorates, so a model that explains the data poorly pays for it in the same currency as one that is merely long. This is also where a class that is misspecified for $F$, in the sense of Section~\ref{sec:restricted}, registers its cost: it is not ruled out in advance but charged for the predictions it gets wrong. The preferred class is the one with the shorter total message, whatever the size of its catalog of possible functions.

Equation~\eqref{eq:MDL} avoids three asymmetries common in informal comparisons. First, both structural and dynamical complexity are charged. Second, graph and hypergraph rules must be regularized at comparable strength. Third, exact emulation is not required: a simpler approximate model may beat a perfect but costly lookup table. Bayesian model comparison is closely related when code lengths are induced by priors, although the marginal likelihood integrates over parameters whereas the two-part expression above minimizes over an encoded model.

The result necessarily depends on the declared code families. We do not see this as a defect, since scientific claims always refer to a hypothesis class, stated or not. The advantage of Eq.~\eqref{eq:MDL} is that it exposes the assumptions: locality, symmetry, maximum interaction order, sparsity, parameter sharing, smoothness, and the treatment of node labels all contribute to the message length.

The outcome of Eq.~\eqref{eq:MDL} depends on how rules are encoded, a choice that is rarely stated. We therefore make it explicit with three codes for the rule of node $i$. The simplest is a lookup table, which lists the output for every configuration of the node and its neighbors; it costs $|\mathcal{X}|^{|\partial i|+1}\log_2|\mathcal{X}|$ bits and is the benchmark that any structured code has to beat. The second writes the rule in the multilinear expansion of Eq.~\eqref{eq:multilinear} and keeps only terms of order at most $d$. Of the $P_{i,d}=\sum_{m\leq d}\binom{|\partial i|+1}{m}$ subsets of $\{i\}\cup\partial i$ that could carry such a term, let $s_i$ be the number with a nonzero coefficient. The code first names these $s_i$ subsets, at a cost of $\lceil\log_2\binom{P_{i,d}}{s_i}\rceil$ bits, and then states their coefficients, at the standard two-part precision of $\tfrac12\log_2 n$ bits each for $n$ observed transitions~\cite{rissanen1978modeling,Grunwald2007}. The third is a hyperedge-aggregated rule with a shared aggregation function, which costs $O(1)$ for the aggregation plus one parameter per incident hyperedge, or a lookup table per hyperedge if the functions are arbitrary, and in either case pays nothing to name subsets, because the structural term has named them already. This contrast is the structure--rule trade-off in practice. The graph rule pays $\lceil\log_2\binom{P_{i,d}}{s_i}\rceil$ for its freedom, where the hypergraph pays $\Len_{\mathcal{L}}(H)$. Parameters shared across nodes are encoded once, which is where a homogeneous mechanism saves bits; and where a two-part code is unsatisfying, the normalized maximum likelihood or a Bayesian mixture over the same family may be substituted, since only differences of code lengths enter.

For an empirical comparison, we recommend reporting at least four quantities:
\begin{enumerate}[label=(\roman*)]
\item the structural length $\Len_{\mathcal{L}}(S)$;
\item the conditional rule length $\Len_{\mathcal{C}}(D\mid S)$;
\item predictive log loss on held-out trajectories or interventions;
\item the stability of the selected interaction order under alternative reasonable codes and data resolutions.
\end{enumerate}

The first two are worth reporting separately even though only their sum enters Eq.~\eqref{eq:MDL}. Moving cost between structure and rule does not change the total and therefore cannot by itself create a preference. The split serves as a diagnostic rather than as an additional criterion, showing whether an observed preference between classes originates in the structural term or in the rule term.

\section{Five regimes}
\label{sec:regimes}
The following examples are elementary on purpose. Rather than universal theorems about graphs or hypergraphs, they illustrate how the sign of a restricted description-length comparison changes with the system and with the admissible codes. Their settings differ, and we state them at the outset: Sections~\ref{sec:ex-group}, \ref{sec:ex-unresolved} and~\ref{sec:ex-lift} treat deterministic Boolean systems, the last rewriting the first of them without hypergraphs, while Section~\ref{sec:ex-diffusion} treats a continuous linear system and Section~\ref{sec:ex-ecology} a parameterized family of continuous population models.

\subsection{A pairwise process}
\label{sec:ex-diffusion}
Consider linear diffusion
\begin{equation}
 \dot{\bm{x}}=-\gamma L_G\bm{x}
 \label{eq:diffusion}
\end{equation}
on a sparse graph, with $L_G$ the graph Laplacian. A graph code stores the edge list and the scalar parameter $\gamma$; the rule is shared across all nodes. If the higher-order language allows size-two hyperedges, the corresponding hypergraph description is a fixed relabeling and the two codes are equivalent up to constants. A hypergraph class that prohibits size-two edges may instead be misspecified, or may require a more elaborate translation, depending on its rule family. In neither case does higher-order vocabulary buy compression or predictive improvement.

Random walks, the discrete-time counterpart of linear diffusion, make the same comparison with a less artificial pair of candidates, and the reduction can be exhibited directly. Let a walker at node $i$ choose uniformly one of the $d_i=\bigl|\{e\in\mathcal{E}:i\in e\}\bigr|$ hyperedges incident to $i$, and then uniformly one of that hyperedge's $|e|-1$ remaining members. Setting
\begin{equation}
 w_{ij}=\sum_{e\in\mathcal{E}:\,\{i,j\}\subseteq e}\frac{1}{|e|-1}
 \label{eq:hyperweights}
\end{equation}
gives $\sum_{j\neq i}w_{ij}=d_i$ whenever every hyperedge has $|e|\geq2$, so the transition probabilities are $P_{ij}=w_{ij}/d_i$: the process is a random walk on the weighted clique projection.

Neither the weights~\eqref{eq:hyperweights} nor the reduction is new. The same construction appears in Ref.~\cite{Peixoto2026} as evidence that hypergraph diffusion is a weighted-graph diffusion, a reading we qualify below, and Ref.~\cite{Llabres2026} proves the analogous exact projection for a class of social impact models, where the projected weights are state-independent only for a linear impact function. Other definitions of a hypergraph walk yield different weights, but the structural conclusion below is unchanged.

The two descriptions can be compared directly. The hypergraph code stores the hyperedge list and a constant-length rule; the weighted-graph code stores the projection together with one weight per supported pair, at whatever precision is demanded. Neither wins in general. For a few disjoint groups the hyperedge list is short and the weights are redundant, so the hypergraph gives the shorter description; in the opposite limit, where heavily overlapping groups project onto a dense, nearly uniform weight matrix, the weighted graph is short while the hyperedge list is long. The sign is set by the overlap structure and the precision of the weights, not by the choice of language.

More interesting is what the comparison leaves open. The weights are not free parameters: by Eq.~\eqref{eq:hyperweights} they are functions of the group structure. A reader given only the weighted graph sees a list of numbers with no account of where they came from, and cannot predict how they would change if a group were added or dissolved. A reader given the hypergraph sees the groups and the rule that produces the weights. The two descriptions may be close in length, or the graph may even be shorter, while differing in what they make explicit and in what they support under intervention.

The example is useful because it keeps the framework from building in a preference for hypergraphs; a pairwise, shared linear mechanism should stay pairwise unless additional structure pays for itself. It also serves as a caution in the other direction. The verdict on Eq.~\eqref{eq:diffusion} turns on description length alone; the verdict on the random walk turns on description length and mechanistic adequacy, and these need not agree.

\subsection{A group-sensitive process}
\label{sec:ex-group}
Let $H$ contain groups of size $k$ and consider the shared Boolean rule
\begin{equation}
 x_i(t+1)=\bigvee_{e\ni i}\;\bigwedge_{j\in e\setminus\{i\}}x_j(t).
 \label{eq:unanimity}
\end{equation}
The hypergraph code lists the groups once and names a constant-length rule. The clique projection $G=\pi(H)$ can emulate the map, but restricted graph rules need not: a conjunction over $k-1$ neighbors is not additively separable for $k>2$, and although a single group's unanimity is a threshold function, a disjunction over two or more groups is in general not linearly separable, even when the groups are disjoint\footnote{For two disjoint groups $\{i,a,b\}$ and $\{i,c,d\}$, suppose $(x_a\wedge x_b)\vee(x_c\wedge x_d)$ were $\ind[\,w_ax_a+w_bx_b+w_cx_c+w_dx_d\geq\theta\,]$. The positive inputs $(1,1,0,0)$ and $(0,0,1,1)$ give $w_a+w_b\geq\theta$ and $w_c+w_d\geq\theta$, the negative inputs $(1,0,1,0)$ and $(0,1,0,1)$ give $w_a+w_c<\theta$ and $w_b+w_d<\theta$, and the two pairs sum to contradictory bounds on $w_a+w_b+w_c+w_d$.}. A more flexible graph class must therefore either encode the operative subsets inside the rule or pay predictive loss.

This cost can be computed, and the result goes against the intuition that a projection is a cheaper description.
\begin{proposition}[Cost of the projection]
\label{prop:projection-cost}
Let $H$ consist of $M$ pairwise disjoint hyperedges of size $k$ on $N$ nodes, so that $H$ is recoverable from $G=\pi(H)$ and $k$. Encode $H$ by naming its hyperedges among the $k$-subsets of $V$, and $G$ by naming its edges among the pairs:
\begin{equation}
 \Len(H)=\Bigl\lceil\log_2\tbinom{\binom{N}{k}}{M}\Bigr\rceil,
 \qquad
 \Len(G)=\Bigl\lceil\log_2\tbinom{\binom{N}{2}}{M\binom{k}{2}}\Bigr\rceil.
 \label{eq:projection-cost}
\end{equation}
Then $\Len(G)/\Len(H)\to k-1$ as $N\to\infty$ at fixed $M$ and $k$.
\end{proposition}
\begin{proof}
For $E\ll Q$ one has $\log_2\binom{Q}{E}=E\log_2 Q-\log_2E!+O(E^2/Q)$, and $\log_2\binom{N}{m}=m\log_2N-\log_2m!+o(1)$. The leading terms are therefore $\Len(H)=Mk\log_2N+O(1)$ and $\Len(G)=2M\binom{k}{2}\log_2N+O(1)=Mk(k-1)\log_2N+O(1)$.
\end{proof}

Since both codes determine the same object, the comparison is not confounded by a difference in what is described. The projection encodes the grouping redundantly, naming $\binom{k}{2}$ pairs where one $k$-set suffices, and the redundancy grows with the order. Three qualifications are in order. The proposition concerns these two codes and has no counterpart at the level of Kolmogorov complexity $\K$, where each object is computable from the other and the lengths agree up to a constant (Section~\ref{sec:projection-cost}); a graph code that names cliques rather than edges avoids the factor, but it is then naming groups, in the sense of Proposition~\ref{prop:incidence}; and if the dynamics never uses the grouping, the graph description need not recover $H$ at all and pays nothing, as in Section~\ref{sec:ex-diffusion}. Ref.~\cite{Peixoto2026} gives a rule-side counterpart, showing that under uniform rule codes and for sparse hyperedges that do not overlap, an admissible hypergraph rule is shorter than the equivalent graph rule on the projected graph. Both separations depend on the declared codes, which is the point of Section~\ref{sec:mdl}.

For a fixed $G$, let $\mathcal{F}_G\subseteq\pi^{-1}(G)$ be the declared family of admissible hypergraphs and $q(\cdot\mid G)$ a declared conditional distribution over it, expressing which groupings are a priori more plausible. Identifying a member costs approximately $-\log_2 q(H\mid G)$ bits, and in the uniform case
\begin{equation}
 \Len_{\mathrm{unif}}(H\mid G)
 =\left\lceil\log_2\left|\mathcal{F}_G\right|\right\rceil.
 \label{eq:fiber}
\end{equation}
The family can be large: even the co-occurrence matrix, which records how many hyperedges contain each pair, generally fails to determine the hypergraph~\cite{LaRockLambiotte2025}, and the unweighted projection retains less.

This is a code choice, not a representation-independent identity. If the skeleton is fixed and carries projected pairwise relations the rule does not otherwise need, the hypergraph may give the shorter restricted description. Equation~\eqref{eq:fiber} must not be added to a full independent code for $H$, which would double-count the grouping; with an optimal unrestricted conditional code, Eq.~\eqref{eq:no-double-count} already bounds the graph description by the hypergraph one up to a constant. The hypergraph-preferred regime therefore comes from the declared graph rule class, the fixed-skeleton constraint, or the inferential cost of learning the subsets, and not from projection ambiguity alone. Parity gives a cleaner illustration, because it separates cost from expressiveness. Suppose the state of a target depends on the parity of a specified $k$-node group. The hypergraph description is short, since the group is named once and the rule is constant-length, but shortness does not decide the comparison, because by Section~\ref{sec:compressibility} parity is highly compressible in either language. The deciding fact is that every graph rule family whose interactions are capped below order $k$ is \emph{misspecified} for this map and, by Eq.~\eqref{eq:MDL}, pays for that in predictive loss rather than in structural bits. The mismatch is between the interaction order of the mechanism and the order the admissible rules can express. This example turns on description length together with the expressive limits of the declared rule class, that is, on functional representability relative to a restriction rather than in the unrestricted sense of Proposition~\ref{prop:emulation}.

\subsection{A contrast from ecology}
\label{sec:ex-ecology}
The distinction between grouping that must be specified independently and grouping that is derived separates published models, not only toy constructions. Bairey, Kelsic, and Kishony add higher-order coefficients to the species growth rates as independent parameters~\cite{Bairey2016}; a weighted graph does not determine those tensors, so a graph rule that reproduces the dynamics must encode the same information. Grilli, Barab\'as, Michalska-Smith, and Allesina, by contrast, study competitive networks in which the higher-order coefficients are completely defined by the pairwise ones~\cite{Grilli2017}, so a weighted graph and a fixed construction generate the higher-order terms at no additional cost. Similar ecological language can therefore carry different representational costs, and hence different preferences: the answer depends on whether higher-order coefficients are independent data or functions of pairwise parameters. Deciding which case applies, model by model, is what Eq.~\eqref{eq:restricted} is meant to enforce. Both verdicts depend on mechanistic adequacy as much as on cost, because what changes between the two descriptions is which quantities are presented as independent inputs and which are derived.

\subsection{An unresolved projection}
\label{sec:ex-unresolved}
Return to $H_1$ and $H_2$ in Eq.~\eqref{eq:H12}. Suppose only the projected $K_4$ and a trajectory are observed. The two rules in Eqs.~\eqref{eq:group-functions-1}--\eqref{eq:group-functions} differ only when a node has exactly two active neighbors, and the states in which none does are closed under both dynamics. A trajectory confined to them is therefore consistent with both rules however long it runs, while any other state separates them at the first update. The set is small and can be written down. Writing $m$ for the number of active nodes, node $i$ has $m-x_i$ active neighbors, so no node has exactly two precisely when $m\in\{0,1,4\}$, which is six of the sixteen states. The failure of identification here is a property of a small set of initial conditions rather than of the projection, and it is worth keeping apart from the structural ambiguity of Lemma~\ref{lem:union}, which no amount of trajectory data removes. It is the model, and not the structure, that fails to be identified. The projection is observed, a hypergraph selected from the clique would be an imputation, and an arbitrary graph lookup table would be an equally under-constrained explanation. Additional observations can change the result. Direct records of four-person events support $H_1$ as structural data. Records covering all four three-person groups support $H_2$. Interventions that set $(x_2,x_3,x_4)=(1,1,0)$ distinguish the two dynamics outright. Both rules depend on the neighborhood only through the number of active neighbors, so the discriminating intervention changes that number rather than regrouping the same neighbors, and a graph rule carrying a single threshold value compresses better than either hyperedge list; group structure that survives count-preserving perturbations is the case of Section~\ref{sec:ex-group}. This regime turns on statistical identifiability, and illustrates why interaction data, state trajectories, and interventions play different inferential roles.

\subsection{A case decided by structural representation}
\label{sec:ex-lift}
The examples so far are decided by cost, by the expressive limits of a declared rule class, or by identifiability. Structural representation, the first of the four questions, decides a case of its own, and it is the case most often used to argue that the choice of language is immaterial. Rewrite the group-sensitive process of Section~\ref{sec:ex-group} twice without hypergraphs: as a factor graph, with one auxiliary node per group joined to its members, and as an edge-colored multilayer graph, in which each hyperedge's clique is placed in a layer of its own~\cite{Peixoto2026}. Both are graphs in the ordinary sense and both reproduce the node-level map exactly: the multilayer graph with the rule taking the conjunction within each layer and the disjunction across layers, and the factor graph provided auxiliary updates count as instantaneous rather than as a time step of their own, since otherwise the lift changes the trajectory and the equivalence does not hold at the level claimed for it.

Three of the four questions then tie, with one qualification. Functional representability ties by construction. Statistical identifiability ties as well: the node-level trajectories agree whenever the auxiliary variables go unobserved, and the description lengths agree up to a constant at the algorithmic level of Section~\ref{sec:projection-cost}, since the group family and the labeled decomposition are computable from each other. The qualification is that under declared codes the factor graph still ties, its incidences costing what the hyperedge list costs, while the layer-labeled clique expansion pays the factor of Proposition~\ref{prop:projection-cost}. Mechanistic adequacy ties too, insofar as both descriptions organize the rule in the same way, as a conjunction within each group and a disjunction across groups. Beyond that, what differs is only which entities the formalism posits as primitive. An auxiliary node is an object: it can be added, removed, or measured, and a model containing one commits itself to a group being that kind of thing, whereas a layer index is a label on edges and commits to nothing beyond which edges belong together. Where the groups are independently identifiable (a complex, a committee, a reaction), the factor and hyperedge descriptions therefore assert what the coloring does not, and between them the choice is settled with cost, fit, and identifiability all silent. In this sense the bipartite equivalence, so often invoked here, is less decisive than it looks. It converts a hypergraph into a graph, but only by naming the groups as nodes. This case shows that two descriptions can agree on everything a data set at the claimed level can test and still differ in what they assert. The factor description commits to groups as entities that could be measured or manipulated, and once such measurements are made the comparison leaves that level and becomes one of identifiability.

\section{Empirical evidence and scientific adequacy}
\label{sec:empirical}
The expressiveness argument carries a strong empirical lesson, which applies to both sides. If every candidate hypergraph model has an exact representative in an unrestricted graph class, trajectories cannot favor the hypergraph language by fit alone. The same fact, however, prevents fit from supporting the unrestricted graph class. A class that can reproduce every map makes no falsifiable prediction beyond the dependency structure it declares, and none at all where the skeleton is itself free (Section~\ref{sec:restricted}).

\subsection{Evidence and inference}
\label{sec:evidence}
The classification that follows deliberately overlaps with the review of claimed evidence in Ref.~\cite{Peixoto2026}, with which we largely agree. Where we differ is in applying the same standard to the unrestricted graph class, which, by the argument above, does not meet it either.

\paragraph{Joint-event data.} Co-authorships, conversations, biochemical complexes, ecological assemblages, and other group events can make a hypergraph the most direct encoding of the observations. This justifies the structure as a data representation, but does not by itself prove that the dynamics responds irreducibly to the whole group.
\paragraph{Dynamical data.} Time series can test restricted rules. Evidence for a higher-order mechanism requires that a regularized group-based model outperform comparably flexible pairwise alternatives in compression or out-of-sample prediction. Recovering hyperedges from a clique projection or a node co-occurrence matrix is in general underdetermined~\cite{LaRockLambiotte2025,YoungPetriPeixoto2021}. Proposition~\ref{prop:emulation} makes this unsurprising: a projected graph does not determine which compatible hyperedge configuration generated the observations, and neither does a trajectory, since a graph can emulate the map exactly.

\paragraph{Constructive precedent.} None of this counts against principled hypergraph inference. Bayesian reconstruction can treat a hypergraph as a latent explanation of pairwise network data and compare it with simpler alternatives, as noted above~\cite{YoungPetriPeixoto2021}; probabilistic hypergraph models can support missing-hyperedge prediction and overlapping-community inference~\cite{ContiscianiBattistonDeBacco2022}; generative hypergraph blockmodels provide model-based clustering objectives~\cite{ChodrowVeldtBenson2021}; and a description-length criterion can already say which orders of a measured hypergraph are structurally redundant~\cite{Kirkley2025}. These works do not all implement Eq.~\eqref{eq:MDL}, but they illustrate the general distinction between inferential models that expose and penalize their assumptions and heuristic imputation from a projection. The line to draw therefore runs between inference that accounts for the structure it posits and inference that does not, whichever language it is written in.
\paragraph{Interventions.} Some perturbations discriminate where observation cannot. Consider holding fixed, for each node, the number of active neighbors while changing which group those neighbors belong to. A rule that responds to intact groups can change its output under such a perturbation, whereas one that responds only to aggregate exposure cannot: the disjoint-group rule of Section~\ref{sec:ex-group} changes it, while the all-triples model $H_2$ of Eq.~\eqref{eq:H12} does not, being symmetric and hence a function of the number of active neighbors alone. A discriminating outcome therefore certifies group sensitivity, whereas a null outcome does not certify aggregate response. Observational equivalence need not survive intervention, which is why interventional data can settle questions that observation alone may not.
\paragraph{Mechanistic measurements.} Sometimes the interaction unit is observed independently of the response: a chemical complex binds jointly, a committee acts collectively, or a synaptic mechanism integrates a specific combination of inputs. Such evidence should enter the structural code or prior rather than being forced to reappear as an opaque fitted rule. It also breaks ties that description length cannot. When an appropriate two-way translation exists, Eq.~\eqref{eq:no-double-count} and its converse leave the graph and hypergraph descriptions carrying the same total algorithmic information up to a constant, yet they need not make the same hypothesis transparent, and which of them does depends on what was measured. If group interactions are directly observed, placing them in $H$ makes the correspondence between data, parameters, and mechanism explicit; if only dyadic transmission is measured, a graph may play that role.

\subsection{The level at which equivalence is claimed}
\label{sec:level}
Phenomenology alone is the weakest form of evidence. Discontinuous transitions, hysteresis, oscillations, and chaos are compatible with multiple microscopic mechanisms. More importantly, an equivalence claim holds only relative to a level of description, and the level is rarely stated. Equality under a homogeneous mean-field closure establishes equivalence at that approximation level and does not by itself establish more: local correlations, overlap, finite-size fluctuations, nucleation pathways, metastable lifetimes, and the response to localized perturbations may all still differ. A claim of structural equivalence should therefore specify whether it concerns a fixed point, a bifurcation diagram, the full stochastic process, or an interventional response.

Triadic percolation makes the distinction concrete~\cite{SunRadicchiKurthsBianconi2023}. In that model a structural graph is coupled to a signed regulatory network whose nodes switch the links between other nodes on or off. The order parameter obeys an iterated map $R^{(t)}=h_p(R^{(t-1)})$, and whether a stationary state is ever reached is governed by the derivative of $h_p$ there (a fixed point $R^{\ast}$ is locally attracting when $|h_p'(R^{\ast})|<1$) and not by the stationary equation $R^{\ast}=h_p(R^{\ast})$, since maps with identical fixed points can differ in which of them are approached. Reproducing that equation is therefore strictly weaker than reproducing $h_p$, which also fixes stability and orbits. The gap matters here, because period doubling and a route to chaos are among the model's central results. Any two models agreeing only on stationary states may thus disagree on the dynamics those states were meant to summarize. A separate observation points in the same direction: the hypergraph generalization in which nodes regulate hyperedges reproduces the same route to chaos~\cite{SunBianconi2024}, so this phenomenology does not discriminate between the two representations either.

\section{Discussion}
\label{sec:discussion}

The observation that a graph can support arbitrary multivariate node functions has changed the terms of the graph--hypergraph debate, and the higher-order literature should take its consequences seriously. A hypergraph toy model does not show that a phenomenon requires a hypergraph, nor does a clique observed in pairwise data identify a hyperedge, and an abrupt transition is not a structural fingerprint. Higher-order models should be compared with graph-based alternatives rather than assumed from phenomenology. The observation does not, however, settle the choice of model. An emulation that relocates the grouping into the rule has changed where the information sits without removing it, and by Proposition~\ref{prop:incidence} the resulting model carries an incidence structure whether or not it is displayed as one. Proposition~\ref{prop:invariance} places this on a footing independent of the parameterization, since the essential interaction sets of the map, and with them its interaction order, are the same in every exact representation. The grouping can be moved, but it cannot be eliminated.

Representational complexity measures what this relocation costs, and the result can favor either side. That the clique projection is not invertible does not by itself favor hypergraphs: an efficient graph code needs to store only the grouping information that the projection leaves undetermined. Nor does it favor graphs that they can represent everything a restricted hypergraph family can: a larger class needs more bits to specify its rule, and with finite data its extra flexibility risks overfitting. Where one language is merely a relabeling of the other, there is no preference to find (Section~\ref{sec:restricted}). That graph-preferred, hypergraph-preferred and unresolved regimes all remain possible is what one should expect of a fair criterion, since a criterion that always selected the same representation would have built that preference in from the start.

This perspective also clarifies the role of alternative constructions. Bipartite incidence graphs, factor graphs, multilayer lifts, and state-space augmentations can faithfully represent higher-order systems, and may be excellent modeling tools. As Section~\ref{sec:ex-lift} shows, however, they need not be neutral between the hypotheses at issue, and by Proposition~\ref{prop:incidence} a construction that introduces one auxiliary node per group has named groups as primitives whatever it is called. A many-to-one projection or a higher-dimensional lift can establish a realization, but not necessarily an invertible, cost-preserving, or mechanistically neutral equivalence.

The practical next step is empirical. A first step in this direction is the learned effective order of Ref.~\cite{Neuhauser2024}, obtained by fitting hypergraph dynamical models of increasing order to trajectories and selecting the lowest order supported by a complexity-penalized score, although within a single hypergraph class and with the hypergraph given rather than encoded. For the same transition data, one can fit comparably regularized graph and hypergraph model families, encode both structure and rule, and compare their held-out message lengths. Synthetic benchmarks should include systems built to be pairwise, systems built to depend on groups, and ambiguous cases, with noise and partial observation varied independently. Empirical studies should add interventions or direct event data where possible. Such comparisons would move the discussion away from universal statements about vocabulary and toward reproducible claims about concrete systems.

\section{Conclusions}
\label{sec:conclusions}
The following two statements do not contradict each other. Graphs with unrestricted node functions can emulate a broad class of higher-order dynamics, and higher-order representations can still be the more parsimonious or scientifically adequate models under explicit restrictions. The first is a result about functional containment; the second is a question of model selection. Emulating higher-order dynamics on a graph is not the same as eliminating higher-order information, and equal expressiveness does not confer equal scientific standing. Two of our results make this concrete. The interaction order of a map (for binary states, exhibited by its essential interaction sets) is an invariant of the map itself, so it is a property of the mechanism and not of the language chosen to display it. And under declared structural codes, describing $M$ disjoint groups of size $k$ through their clique projection costs asymptotically $k-1$ times as much as listing the groups, so encoding the grouping through pairs is the more expensive choice whenever the dynamics uses it. No structural language is universally privileged when dynamical rules are unrestricted. Pairwise and higher-order representations become scientifically distinguishable relative to constrained model classes that are not nested in one another and informative data, where structural representation, description length, prediction, identifiability, and mechanistic adequacy can be compared, and reported as complementary criteria instead of being collapsed into a single ranking. The debate can be resolved at this level, whereas comparisons between unrestricted encodings leave it open.

\section*{Acknowledgments}
This work was supported by the European Union's Horizon Europe Marie Sk\l odowska-Curie Actions under the ``BeyondTheEdge: Higher-Order Networks and Dynamics'' project (Grant Agreement No.~101120085). G.F.A. was supported by the Funda\c{c}\~ao de Amparo \`a Pesquisa do Estado de S\~ao Paulo (FAPESP), Process Nos.~2024/16711-8 and 2025/04409-8. Y.M. was partially supported by the Government of Arag\'on, Spain, and ERDF ``A way of making Europe'' through grant E36-23R (FENOL), and by Grant No.~PID2023-149409NB-I00 from MICIU/AEI (10.13039/501100011033) and ERDF ``A way of making Europe.''

\end{document}